\documentclass[ twocolumn, notitlepage,superscriptaddress, amsmath,amssymb, aps,pra]{revtex4-1}
\usepackage{tikz}
\usepackage{amscd}
\usepackage{amsthm}
\usepackage{graphicx}
\usepackage{mathdots}
\usepackage{epstopdf}
\usepackage{enumerate}
\usepackage{changepage}

\usepackage[
colorlinks,
linkcolor = blue,
citecolor = blue,
urlcolor = blue]{hyperref}
\def \qed {\hfill \vrule height7pt width 7pt depth 0pt}

\newtheorem{theorem}{Theorem}
\newtheorem{corollary}{Corollary}
\newtheorem{proposition}{Proposition}
\newtheorem{lemma}{Lemma}
\newtheorem{example}{Example}

\usepackage{dcolumn}
\usepackage{bm}
\usepackage{bbm}

\def\<{\langle}
\def\>{\rangle}

\DeclareMathAlphabet\mathbfcal{OMS}{cmsy}{b}{n}

\begin{document}

\title{Beyond the mixture of  generalized    Pauli  dephasing channels }

\author{Mao-Sheng Li}
\email{li.maosheng.math@gmail.com}
\affiliation{ School of Mathematics,
	South China University of Technology, Guangzhou
	510641,  China}

\author{Wen Xu}
 
\affiliation{ School of Mathematics,
	South China University of Technology, Guangzhou
	510641,  China}

\author{Yan-Ling Wang}
\email{wangylmath@yahoo.com}
\affiliation{ School of Computer Science and Technology, Dongguan University of Technology,Dongguan
	523808, China }

\author{Zhu-Jun Zheng}
\email{zhengzj@scut.edu.cn}
\affiliation{ School of Mathematics,
	South China University of Technology, Guangzhou
	510641,  China}

\begin{abstract}
	 In recent times, there has been a growing scholarly focus on investigating the intricacies of quantum channel mixing.
	  It has been commonly believed, based on intuition in the literature, that every generalized Pauli channel with dimensionality $d$ could be represented as a convex  combination of $(d+1)$ generalized Pauli dephasing channels (see [\href{https://journals.aps.org/pra/abstract/10.1103/PhysRevA.103.022605} {Phys. Rev. A \textbf{103}, 022605 (2021)}] as a reference). To our surprise, our findings indicate  the inaccuracy of this intuitive perspective. This has stimulated our interest in exploring  the properties of convex combinations of generalized Pauli channels, beyond the restriction to just $(d+1)$ generalized Pauli dephasing channels.  	We demonstrate that many previously established properties still hold within this broader context.  For instance, any mixture of invertible generalized Pauli channels retains its invertibility. It's worth noting that this property doesn't hold when considering the Weyl channels setting. Additionally, we demonstrate that every Pauli channel (for the case of $d=2$) can be represented as a mixture of $(d+1)$ Pauli dephasing channels, but this generalization doesn't apply to higher dimensions. This highlights a fundamental distinction between qubit and general qudit cases. In contrast to prior understanding, we show that non-invertibility of mixed channels is not a prerequisite for the resulting mapping to constitute a Markovian semigroup.

\end{abstract}      

\flushbottom

\maketitle

\thispagestyle{empty}

 \section{Introduction}

 The temporal evolution of a quantum system is elegantly described  through a family of completely positive and trace-preserving (CPTP) mappings denoted as $\Lambda_t ~(t\geq0)$, satisfying the initial condition $\Lambda_0=\mathbbm{1}$. These dynamical mappings are commonly referred to as quantum channels. Given an initial state $\rho$ of the system, the state $\rho_t= \Lambda_t(\rho)$ characterizes the evolution of $\rho$.
 
 In the context of a closed quantum system, the evolution of quantum states is governed by unitary transformations, which can be succinctly expressed as $\Lambda_t(\rho) = U_t \rho U_t^\dag$, where $U_t=e^{-\mathrm{i}Ht}$ with $H$ signifying the Hamiltonian of the enclosed system (with $\hbar=1$ set as a convention).  In contrast, when we shift our focus to open quantum systems \cite{Bre2002,Wei2000,Riv2011}, interactions between the system and an external environment introduce nontrivial effects, leading to phenomena such as dissipation, decay, and decoherence \cite{Riv2014,Bre2016,Veg2017}. Consequently,   the dynamical evolution no longer strictly follows the principles of unitarity. In such scenarios, it is customary to employ a suitable Born-Markov approximation to describe the evolution. This approximation is governed by the Markovian semigroup master equation
$$\dot{\Lambda}_t=\mathcal{L}\circ\Lambda_t,$$
where $\mathcal{L}$ is the time-independent generator given by
\begin{equation}\label{2}
\mathcal{L}(\rho)=-\mathrm{i}[H,\rho]+\sum\limits_{\alpha}\gamma_{\alpha}\left(V_{\alpha}\rho V_{\alpha}^{\dag}-\frac{1}{2}\{V_{\alpha}^{\dag}V_{\alpha},\rho\}\right),
\end{equation}
with $V_{\alpha}$ being the noise operators, and $\gamma_{\alpha}\geq0$ representing the local decoherence rates \cite{Gor1976,Lin1976}. Solving this master equation yields the dynamical map as $\Lambda_t=e^{t\mathcal{L}}$, ensuring that it remains a legitimate CPTP dynamical map.   A crucial feature of Markovian semigroup dynamics is that $\Lambda_{s+t}=\Lambda_s \circ \Lambda_t$ for all $s,t \geq 0.$  Going beyond the constraints of the Markovian semigroup, we frequently encounter scenarios where the generator $\mathcal{L}$ becomes time-dependent $\mathcal{L}_t$.    It has exactly the same form as (\ref{2}) with time-dependent $\mathcal{H}(t), V_{\alpha}(t)$,  and  $\gamma_{\alpha}(t)$. The formal solution for $\Lambda_t$ is
\begin{equation}\label{formal_solution}
\Lambda_t=\mathcal{T}\exp{\left(\int_{0}^{t}\mathcal{L}_\tau\mathrm{d}\tau\right)},
\end{equation}
where $\mathcal{T}$ is the chronological time-ordering operator \cite{Riv2010}. 
It is crucial to note that such a dynamical map is significantly more complex, and we currently lack a complete understanding of the necessary and sufficient conditions for $\mathcal{L}_t$ that ensure (\ref{formal_solution}) remains a CPTP map for all $t\geq0$. Recently, time-dependent generators $\mathcal{L}_t$ have been widely employed to investigate quantum non-Markovian evolutions \cite{Riv2014,Bre2016,Veg2017,Cre2007,Chr2022}.

Dynamical maps $\Lambda_t$ are said to be divisible if they can be represented as
\begin{equation}
	\Lambda_t=V_{t,s}\circ\Lambda_s,
\end{equation}
with $V_{t,s}$ serving as the intermediate map for all $t\geq s$. 
Furthermore, these maps are categorized as P-divisible if $V_{t,s}$ is both positive and trace-preserving (PTP), and CP-divisible if $V_{t,s}$ is completely positive and trace-preserving (CPTP). CP-divisibility of $\Lambda_t$ is frequently employed as a criterion to define Markovianity \cite{Riv2010}. 

Recently, there has been a surge in interest in exploring the properties arising from the convex combination of quantum channels  \cite{Wol2008,Chr2010,Wud2015,Wud2016,Meg2017,Siudzinska_JPA20,Siudzinska_JPA22b,Bre2018,Jag2020-1,Jag2020-2}. Notably, Pauli channels and generalized Pauli channels, two well-established channel types, have been extensively investigated  \cite{Chr2016,Siudzinska21,Utagi21,Jagadish2022,Jagadish_Measure,Siudzinska_JPA22}. Different mixing ways of them may lead to  the emergence of intriguing properties of the resultant maps, such as the  Markovianity,  Markov semigroup structure,   and  singularity \cite{Hou2021}. Several interesting works \cite{Siudzinska21,Siudzinska_JPA22} are predicated on the intuition that every generalized Pauli channel with dimensionality $d$ can be expressed as a convex combination of $(d + 1)$ generalized Pauli dephasing channels. In this work, we challenge this intuition and demonstrate that it does not hold, motivating us to comprehensively investigate the properties arising from the convex combination of general generalized Pauli channels. 

 The paper is organized as follows. In Sec. \ref{second}, we give some notations and introduce the definition of generalized Pauli channels. In Sec. \ref{three}, we study the properties of the resultant maps of a convex combination of   generalized Pauli channels, such as Markovian semigroup, invertibility, subsets relations. Finally, we conclude our findings in Section \ref{four}.

 \section{Preliminaries}\label{second}

 In this paper, we adopt the following notations: the set $[n]$ is represented as ${1,2,\cdots, n}$ for any positive integer $n$. The Hilbert space of dimension $d$ is denoted as $\mathcal{H}_d$. Furthermore, we use $\mathbb{D}_d$ to denote the set of all density matrices associated with the system   $\mathcal{H}_d$, while $\mathbb{L}_d$ is employed to represent the set of all linear operations acting from $\mathcal{H}_d$ to itself.

 A mixed unitary evolution of a qubit is precisely described by the Pauli channel  expressed  by
\begin{equation}\label{Pauli}
	\Lambda_t^{(\mathbf{p})}[\rho]=\sum_{\alpha=0}^3 p_\alpha (t) \sigma_\alpha \rho \sigma_\alpha,
\end{equation}
where ${\mathbf{p}}=(p_0(t),p_1(t),p_2(t),p_3(t))$ and $p_\alpha(t)$ denote the probability distribution with $p_0(0)=1$, and $\sigma_\alpha$ are the Pauli matrices.
This formulation represents the most general form of a bistochastic quantum channel \cite{King,Landau}.

A natural extension of Pauli channels to higher dimensions arises through the concept of mutually unbiased bases (MUBs) \cite{Bandyopadhyay02}. Let $d\geq 2$ be an integer. Two normalized orthogonal bases $\{|\psi_{i}\rangle\}_{i=1}^d$ and  $\{|\phi_{j}\rangle\}_{j=1}^d$ of $\mathcal{H}_d$ are  called  mutually unbiased if $|\langle \psi_i|\phi_j\rangle|^2 =\frac{1}{d}$ for all $i,j\in [d].$  Suppose that there exists $d+1$ MUBs $\mathcal{B}_\alpha=\{|\phi_k^{(\alpha)}\rangle\}_{k=1}^d$ where $\alpha=1,2,\cdots, d+1,$ which have been  known to be held  when $d$ is an integer power of some prime numbers \cite{Bandyopadhyay02}. For each $\alpha \in [d+1], $ we can define    a unitary operator 
\begin{equation}\label{eq:UU}
	U_{\alpha}=\sum_{l=0}^{d-1}\omega_d^{l}|\phi_k^{(\alpha)}\rangle\langle \phi_k^{(\alpha)}|,\qquad \omega_d=e^{\frac{2\pi \mathrm{i}}{d}}.
\end{equation}
Given a time dependent probability distribution $\mathbf{p}=(p_0(t),p_1(t),\cdots,p_{d+1}(t))$ with $p_0(0)=1$ where each $ p_\alpha(t)$ is  also assumed to be a continuous function on $[0,\infty),$ one   defines  the generalized Pauli  channel  as 
\begin{equation}\label{GPC2}
	\Lambda_t^{(\mathbf{p})}[\rho]=p_0(t) \rho+\frac{1}{d-1}\sum_{\alpha=1}^{d+1}p_\alpha(t)\mathbb{U}_\alpha[\rho],
\end{equation}
where
$	\mathbb{U}_\alpha[\rho]=\sum_{k=1}^{d-1}U_{\alpha}^k\rho{U_{\alpha}^k}^\dagger.$
The dynamical map $\Lambda_t^{(\mathbf{p})}$ is also assumed to satisfy the   equation 
$
\dot{\Lambda}_t^{(\mathbf{p})} [\rho] = {\mathcal{L}}_t^{(\mathbf{p})}  \circ {\Lambda}_t^{(\mathbf{p})} [\rho],$
where  
$${\mathcal{L}}_t^{(\mathbf{p})}[\rho]=\frac{1}{d}\sum_{\alpha=1}^{d+1} \gamma_\alpha(t)
\left( \mathbb{U}_\alpha[\rho]-(d-1) \rho\right),$$
which is called time-local generator of $\Lambda_t^{(\mathbf{p})}.$ If all $\gamma_\alpha(t)$ are non-negative constants, $\Lambda_t^{(\mathbf{p})}$ forms a semigroup: $\Lambda^{(\mathbf{p})}_{s+t}=\Lambda^{(\mathbf{p})}_s \circ \Lambda^{(\mathbf{p})}_t$ for all $s,t \geq 0$   and we call it a Markovian semigroup. We  denote $\mathcal{P}_d$   the set of all generalized Pauli channels and $\mathcal{S}_d$ the set of all  Markovian semigroups in $\mathcal{P}_d$.

 Now, the eigenvalue equations for $\Lambda_t^{(\mathbf{p})}$ read $\Lambda_t^{(\mathbf{p})}[\mathbb{I}_d]=\mathbb{I}_d$ and $
	\Lambda_t^{(\mathbf{p})}[U_{\alpha}^k]=\lambda_\alpha(t) U_{\alpha}^k
$
with 
\begin{equation}\label{eq:eiglambda}
\lambda_\alpha(t)=1-\frac{d}{d-1}\left(\sum_{\beta=1}^{d+1} p_\beta(t)-p_\alpha(t)\right)
\end{equation}
for each $\alpha\in[d+1]$  and $k\in [d-1].$

 For each $\alpha\in [d+1]$  and  a differentiable function $\pi_\alpha(t)$ on $[0,\infty)$ with $\pi_\alpha(0)=0$, we can define a probability distribution 
$$ \mathbf{p}_{\alpha, \pi_\alpha}:=(1-\pi_\alpha(t),0,\cdots,0, \pi_\alpha(t),0,\cdots,0),$$
where $\pi_\alpha(t)$ is at the $(\alpha+1)$-{th} coordinate. Set $$\mathcal{D}_{d,\alpha}:=\{\Lambda_t^{(\mathbf{p} )} \in \mathcal{P}_d\mid  \Lambda_t^{(\mathbf{p} )} = \Lambda_t^{(\mathbf{p}_{\alpha,\pi_\alpha})} \text{ for some }  \pi_\alpha(t) \},$$
and we call it the $\alpha$-th  generalized Pauli dephasing channels.  And we define $\mathcal{D}_d$ as the set of  convex combinations of elements in $\cup_{\alpha=1}^{d+1} \mathcal{D}_{d,\alpha}.$ 
Note that $\mathcal{D}_{d,\alpha}$ is a convex set for every $\alpha\in [d+1].$  As a consequence, every element $\Lambda_t^{(\mathbf{p})} \in \mathcal{D}_d$ can be written as 
\begin{equation}\label{eq:decomD}
	\Lambda_t^{(\mathbf{p})}=\sum_{\alpha=1}^{d+1} x_\alpha \Lambda_t^{(\mathbf{p}_{\alpha,\pi_\alpha})},
\end{equation} 
where $x_\alpha\geq 0$ and $\sum_{\alpha=1}^{d+1} x_\alpha=1.$ In fact, suppose that 
\begin{equation}\label{eq:mixdecomD}
	\Lambda_t^{(\mathbf{p})}=\sum_{\alpha=1}^{d+1} \sum_{j=1}^{N_\alpha} x_{\alpha,j} \Lambda_t^{(\mathbf{p}_{\alpha,\pi_{\alpha,j}})},
\end{equation} 
where $x_{\alpha,j}\geq 0$ and $\sum_{\alpha=1}^{d+1} \sum_{j=1}^{N_\alpha} x_{\alpha,j}=1.$ Then we define $x_\alpha=\sum_{j=1}^{N_\alpha} x_{\alpha,j}.$ If $x_\alpha>0$, we define
$$ \pi_{\alpha}(t):=\sum_{j=1}^{N_\alpha} \frac{x_{\alpha,j}}{x_\alpha}  \pi_{\alpha,j}(t), $$ otherwise, define $ \pi_{\alpha}(t):=0.$ Then expression \eqref{eq:mixdecomD} can be rewritten as Eq. \eqref{eq:decomD} under these definitions.

The complete positivity conditions for $\Lambda_t^{(\mathbf{p})}$ are the generalized Fujiwara-Algoet conditions \cite{Fujiwara,Ruskai}
\begin{equation}\label{eq:FA}
	-\frac{1}{d-1}\leq\sum_{\alpha=1}^{d+1}\lambda_\alpha (t)\leq 1+d\min_\alpha\lambda_\alpha(t).
\end{equation}



\noindent
 \section{Properties of  generalized Pauli channels}\label{three}
In this section, we first demonstrate that not every generalized Pauli channel with dimensionality $d$ can be represented as a convex combination of $(d + 1)$ generalized Pauli dephasing channels. Then, we shift our focus to investigating the properties that emerge from the convex combination of general generalized Pauli channels.
\begin{proposition}\label{pro:genPauli}
	Given an integer $d\geq 2$  and a probability distribution $\mathbf{p}=(p_0(t),p_1(t),\cdots,p_{d+1}(t)),$  the generalized Pauli channel $\Lambda_t^{(\mathbf{p})}$ belong to $\mathcal{D}_d$  if and only if 
\begin{equation}\label{eq:lessthan1}
	\sum_{\alpha=1}^{d+1}\left(\sup_{t\geq 0}  p_\alpha(t) \right) \leq 1.
\end{equation}  
	As a consequence, the set $\mathcal{D}_d$ is strictly containing in $\mathcal{P}_d.$ That is, $\mathcal{D}_d \subseteq \mathcal{P}_d$ and there exists $\Lambda_t^{(\mathbf{p})}\in \mathcal{P}_d$ such that $\Lambda_t^{(\mathbf{p})}\notin \mathcal{D}_d.$
\end{proposition}

\begin{proof}

{\bf Necessity.}  Suppose that $\Lambda_t^{(\mathbf{p})}$ belongs to $\mathcal{D}_d$.  As mentioned before,  there exists   $0\leq \pi_\alpha(t)\leq 1$  with $\pi_\alpha(0)=0$ such that 	
$$ \Lambda_t^{(\mathbf{p})}=\sum_{\alpha=1}^{d+1} x_\alpha \Lambda_t^{(\mathbf{p}_{\alpha,\pi_\alpha})},$$
where $x_\alpha\geq 0$ and $\sum_{\alpha=1}^{d+1} x_\alpha=1.$   
So by the Corollary \ref{cor:channelequal} in Appendix \ref{sec:AppendixA}, one should have  
$$
p_\alpha(t)=x_\alpha \pi_\alpha(t), \text{ for } \alpha\in [d+1].$$
Therefore, taking the supremum of both side at each of the above   equations, one has
$$	 
\displaystyle\sup_{t\geq 0} p_\alpha(t)\leq  x_\alpha,  \text{ for } \alpha\in [d+1].$$  
Taking the sum of both side,  one has 
$$ \sum_{\alpha=1}^{d+1}\left(\sup_{t\geq 0}  p_\alpha(t) \right) \leq \sum_{\alpha=1}^{d+1} x_\alpha=1.$$
\noindent	{\bf Sufficiency.} Denote $m_\alpha:=\sup_{t\geq 0}  p_\alpha(t)$ for $\alpha\in[d+1].$
By assumption, we have $\sum_{\alpha=1}^{d+1} m_\alpha\leq 1.$  Therefore, we can suppose that $m_{0}\geq 0$ is the number satisfies $ m_0+ \sum_{\alpha=1}^{d+1} m_\alpha= 1.$ For each $\alpha\in [d+1]$, we define 
$$ \pi_\alpha(t)=\begin{cases}
	\displaystyle	\frac{p_\alpha(t)}{m_\alpha}, &  m_\alpha >0,\\
	0,& m_\alpha=0,
\end{cases}$$
and $\pi_0(t)=0.$ One has $d+2$ generalized Pauli dephasing channels 
$$ \Lambda_t^{(\mathbf{p}_{1,\pi_0})},  \text{ and } \Lambda_t^{(\mathbf{p}_{\alpha,\pi_\alpha})}, \forall \alpha \in [d+1].$$  
It is easy to check that 
$$ \Lambda_t^{(\mathbf{p})}=m_0 \Lambda_t^{(\mathbf{p}_{1,\pi_0})}+\sum_{\alpha=1}^{d+1} m_\alpha\Lambda_t^{(\mathbf{p}_{\alpha,\pi_\alpha})}.$$

For the last statement of the proposition, one notes that it is easy to construct some probability distribution functions whose corresponding equation \eqref{eq:lessthan1}  is violated. 

 \end{proof}

\begin{figure}[ht]
	\centering
	\includegraphics[scale=0.31]{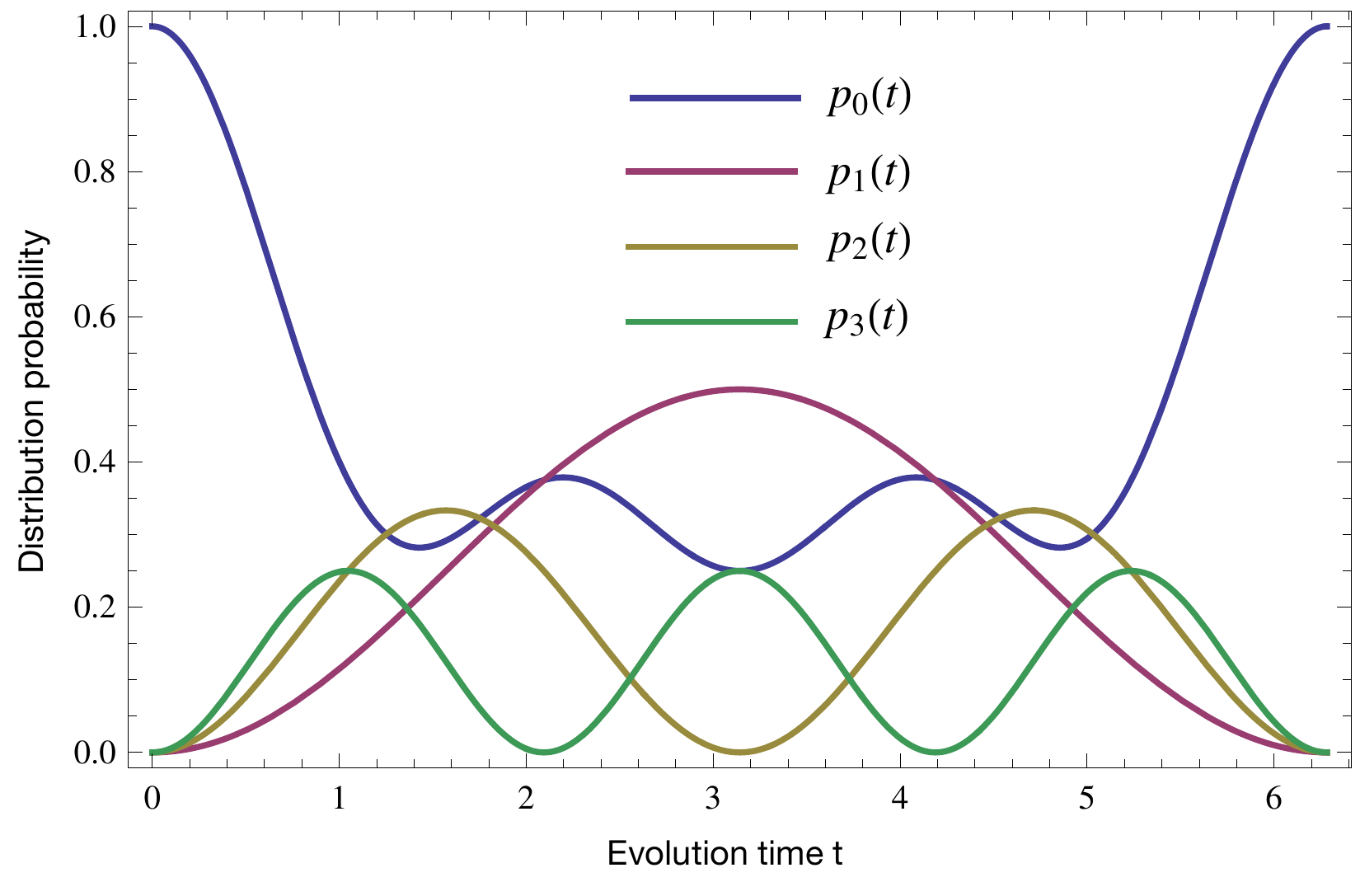}
	\caption{The figure of $p_\alpha(t)$ defined in Eq. \eqref{eq:PauliChannel} for $\alpha=0,1,2,3$ when $t\in[0,2\pi].$ }\label{fig:probability}
\end{figure} 

Here we take $d=2$ for an example  to illustrate that $\mathcal{D}_2$ is strictly containing in $\mathcal{P}_2.$ Suppose that    $p_1(t),p_2(t),p_3(t)$ are defined as 
\begin{equation}\label{eq:PauliChannel}
p_\alpha(t)=\frac{1}{2(\alpha+1)}(\cos [ \alpha t+\pi]+1), \alpha =1,2,3.
\end{equation}
And set $p_0(t)=1-\sum_{\alpha=1}^3 p_\alpha(t)$ and 
$$\mathbf{p}:=(p_0(t), p_1(t),p_2(t),p_3(t)),$$
see  Fig. \ref{fig:probability} for an intuition of the $p_\alpha(t).$ As $p_\alpha(t)\geq 0$ and $\sum_{\alpha=0}^3 p_\alpha(t)=1,$  the map $\Lambda_t^{(\mathbf{p})}$ forms a legitimate Pauli channel.  By the first statement of Proposition \ref{pro:genPauli}, $\Lambda_t^{(\mathbf{p})} \notin \mathcal{D}_2$ as the sum of supremums of $p_\alpha(t)$ is 
$$\frac{1}{2}+ \frac{1}{3}+\frac{1}{4}=\frac{13}{12}>1.$$

Several works \cite{Utagi21,Siudzinska21} have discussed the inability to generate noninvertible channels through convex combinations of $(d+1)$ invertible generalized Pauli dephasing channels. Nevertheless, it is crucial to recognize that this   does not preclude the generation of noninvertible channels through convex combinations of invertible generalized Pauli channels    which leads us to the following proposition.

\begin{proposition}\label{pro:genPauli_invertible} 
	Any mixture of  invertible generalized Pauli channels is still invertible. As a consequence, the invertible generalized Pauli  channels form a convex set. 
\end{proposition}
\begin{proof}
	Suppose that  $\{\Lambda_t^{(\mathbf{p}_k)}\}_{k=1}^K$ are all invertible generalized Pauli  channels and 
	$$ \Lambda_t^{(\mathbf{p})}=\sum_{k=1}^K x_k \Lambda_t^{(\mathbf{p}_k)},  \text{ where }  x_k\geq 0 \text{ and }  \sum_{k=1}^K x_k =1.$$ 
	
	Suppose that $\Lambda_t^{(\mathbf{p}_k)}[U_\alpha] =\lambda_\alpha^{(\mathbf{p}_k)}(t)U_\alpha$ and  $\Lambda_t^{(\mathbf{p})}[U_\alpha] =\lambda_\alpha^{(\mathbf{p})}(t)U_\alpha,$ then we have 
	\begin{equation}\label{eq:Lambda_com}
		\lambda_\alpha^{(\mathbf{p})}(t)=\sum_{k=1}^K x_k \lambda_\alpha^{(\mathbf{p}_k)}(t).
	\end{equation}
	By Eq. \eqref{eq:eiglambda}, the eigenvalues of  $\Lambda_t^{(\mathbf{p}_k)}$ are 
	$$\lambda_\alpha^{(\mathbf{p}_k)}(t)=1-\frac{d}{d-1} \left(\sum_{\beta=1}^{d+1} p_\beta^{(\mathbf{p}_k)}(t)- p_\alpha^{(\mathbf{p}_k)}(t)\right),$$
	which are real continuous functions on $[0,+\infty).$ Moreover, we have $\lambda_\alpha^{(\mathbf{p}_k)}(0)=1.$  As $\Lambda_t^{(\mathbf{p}_k)}$ is invertible,  therefore,   we must have  $ \lambda_\alpha^{(\mathbf{p}_k)}(t)>0$ for all $t\geq 0$ and $\alpha\in [d+1].$  With these relations  and Eq. \eqref{eq:Lambda_com}, one deduces that 
	$\lambda_\alpha^{(\mathbf{p})}(t)>0$ for all $t\geq 0$ and $\alpha\in[d+1].$ That is, $\Lambda_t^{(\mathbf{p})}$ is invertible. 
 
\end{proof}

However, this property does not hold for general  channels. In fact, there exists some invertible  channels whose mixture maybe non-invertible.
For each $(k,l)\in\mathbb{Z}_{d}\times\mathbb{Z}_{d}$, denote $U_{kl}$  be   the unitary Weyl operator,
\begin{equation}
	U_{kl}=\sum\limits_{m\in\mathbb{Z}_{d}}\omega_d^{km}|m\rangle\langle m+l|,~\omega_d=e^{\frac{2\pi \mathrm{i}}{d}}.
\end{equation} 
The time-dependent generalized Weyl channel is defined by
\begin{equation}\label{eq:Weyl}
	\mathcal{E}_t^{(\mathbf{p})}(X)=\sum\limits_{(i,j)\in\mathbb{Z}_{d}\times\mathbb{Z}_{d}}p_{ij}(t)U_{ij}XU_{ij}^{\dag},
\end{equation}
where $\sum\limits_{(i,j)\in\mathbb{Z}_{d}\times\mathbb{Z}_{d}}p_{ij}(t)=1$, $p_{ij}(t)$ is the time-dependent probability distribution such that \text{$p_{00}(0)=1$, $p_{ij}(0)=0$} for $(i,j)\neq(0,0)$, which guarantee that $\mathcal{E}_0^{(\mathbf{p})} =\mathbbm{1}.$

\begin{example}\label{exam:non=invertible}
Let $d=3$ and  \begin{equation}
	\begin{array}{l}
		\displaystyle p_1(t)=q_2(t)=\frac{1-e^{-t}}{2};\\
		\displaystyle p_2(t)=q_1(t)=\frac{1-e^{-t}}{3};\\
		\displaystyle p_0(t)=q_0(t)=1-p_1(t)-p_2(t).
	\end{array}
\end{equation}
Define $\mathbf{p}=(p_{ij}(t))_{(i,j)\in\mathbb{Z}_3\times \mathbb{Z}_3}$ and  $\mathbf{q}=(q_{ij}(t))_{(i,j)\in\mathbb{Z}_3\times \mathbb{Z}_3}$, where
$p_{ij}(t)=p_i(t)p_j(t)$ and $q_{ij}(t)=q_i(t)p_j(t).$ The dynamical maps $\mathcal{E}_t^{(\mathbf{p})}$ and  $\mathcal{E}_t^{(\mathbf{q})}$ 
 defined by Eq. \eqref{eq:Weyl}   are invertible, but their mixture $\frac{1}{2}\mathcal{E}_t^{(\mathbf{p})}+\frac{1}{2}\mathcal{E}_t^{(\mathbf{q})}$ is not.
\end{example}
The proof of Example \ref{exam:non=invertible} is given in Appendix \ref{sec:appB}. 

It is well-established that a generalized Pauli channel constitutes a Markovian semigroup if and only if its local decoherence rates remain nonnegative constants. In this context, we present a definitive and comprehensive condition for a generalized Pauli channel to qualify as a Markovian semigroup, which depends on the spectral properties of the channel.

\begin{theorem}\label{thm:GPauli}
	A   generalized Pauli channel $\Lambda_t^{(\mathbf{p})}$ is a Markovian semigroup if and only if its  spectra $\lambda_\alpha(t)$'s are ${1,e^{-c_1t},  \cdots ,e^{-c_{d+1}t}}$, where $c_\alpha$'s are nonnegative real constants and 
 satisfy $$\sum_{\alpha=1}^{d+1} c_\alpha \geq d\max_\beta\{c_\beta\}.$$
\end{theorem} 

\begin{proof}
	{\bf Necessity.} If the generalized Pauli channel $\Lambda_t^{(\mathbf{p})}$ is a Markovian semigroup, it is clear that its spectra are of the  form
	$$\lambda_\alpha(t)=e^{-c_\alpha t},\ \  c_\alpha\geq 0,$$
	and Eq. \eqref{eq:FA} holds. Therefore,  if for each $\beta\in [d+1]$ we define
	$$ F_\beta(t):=1+d e^{-c_\beta t}-\sum_{\alpha=1}^{d+1} e^{-c_\alpha t},$$
	we should have $F_\beta(t)\geq 0,$ for all $t\geq 0$. As $F_\beta(0)=0,$  we must have  $F_\beta'(0)\geq 0.$ Otherwise, $F_\beta'(0)< 0$ implies there must exist some $t>0$ such that $F_\beta(t)<0.$ Therefore, we have 
	$ F_\beta'(0)=\sum_{\alpha=1}^{d=1} c_\alpha -d c_\beta\geq 0$
	for each $\beta$ which yields our conclusion.
	
	\noindent	{\bf Sufficiency.} Without loss of generality, we assume that $c_1\geq c_2\geq\cdots \geq  c_{d+1}.$ Therefore, $F_{d+1}(t)\geq F_d(t)\geq\cdots\geq  F_1(t).$ Moreover, by our given condition, we have $ \sum_ {\alpha=1}^{d+1} c_\alpha -d c_{1}\geq 0$. Therefore, 
	$$\begin{array}{l}
		F'_{1}(t)	= \displaystyle \sum_{\alpha=1}^{d+1} c_\alpha e^{-c_\alpha t} -dc_{1} e^{-c_{1} t}
		\geq  \displaystyle ( \sum_{\alpha=1}^{d+1} c_\alpha  -dc_1) e^{-c_1 t}\geq 0.\\ [3mm]
	\end{array}
	$$
	Therefore, $F_1(t)$ is an increasing monotone function.  Hence, $F_1(t)\geq F_1(0)=0$ for all $t\geq 0.$ From the above argument we found that $F_\beta(t)\geq 0$ for all $\beta.$ By defining $p_\beta(t)= \frac{d-1}{d^2}F_\beta(t), $   one finds that $\Lambda_t^{(\mathbf{p})}$ is a legitimate generalized Pauli channel with spectra being ${1,e^{-c_1t}, e^{-c_2t},\cdots, e^{-c_{d+1}t}}$.  As the spectra of $\Lambda_t^{(\mathbf{p})}$ are all of  the form $e^{ct}$ with $c$ being constant,  
	it  must be  a Markovian semigroup.
	
\end{proof}

With Theorem \ref{thm:GPauli} at hand, we can generalize the previous known statement: any nontrivial mixture of $(d+1)$  generalized Pauli dephasing channels that are Markovian semigroups is not   a   Markovian semigroup again.

\begin{proposition}\label{pro:genPauli2}
	Any nontrivial convex combination of elements in $\mathcal{S}_d$ must be lying outside of $\mathcal{S}_d.$
\end{proposition}

\begin{proof}

	We will prove the   statement by contradiction. Assume that   $\Lambda_t^{(\mathbf{p}_k)}$, $k=1,\cdots,K$ (where $K\geq 2$) are different Markovian semigroups  and there exists a Markovian semigroup $\Lambda_t^{(\mathbf{p})}$ and  $x_k>0$ with $\sum_{k=1}^K x_k=1$  such that 
	\begin{equation}\label{eq:decom}
		\Lambda_t^{(\mathbf{p})}=\sum_{k=1}^K x_k  \Lambda_t^{(\mathbf{p}_k)}.
	\end{equation} 
	Assume that 
	\begin{equation}\label{eq:eigens}
		\Lambda_t^{(\mathbf{p})}[U_\alpha]=e^{-c^{(\mathbf{p})}_\alpha t} U_\alpha,\ \Lambda_t^{(\mathbf{p}_k)}[U_\alpha]=e^{-c^{(\mathbf{p}_k)}_\alpha t} U_\alpha
	\end{equation}
	for $k=1,2,\cdots,K$ and $\alpha=1,2,\cdots,d+1.$
	By Eqs. \eqref{eq:decom} and  \eqref{eq:eigens}, we should  have 
	$$ e^{-c^{(\mathbf{p})}_\alpha t}=\sum_{k=1}^K x_k  e^{-c^{(\mathbf{p}_k)}_\alpha t},$$
	which implies that 
	\begin{equation}\label{eq:identity}
		1=\sum_{k=1}^K x_k  e^{(c_\alpha ^{(\mathbf{p})}-c^{(\mathbf{p}_k)}_\alpha ) t}.
	\end{equation}
	Clearly,  there is no $k\in [K]$  such that $(c_\alpha^{(\mathbf{p})}-c^{(\mathbf{p}_k)}_\alpha )>0$. Otherwise, taking $t$ tend to infinity, the right hand side of Eq. \eqref{eq:identity} must tend to infinity as $x_k>0$ for all $k\in[K]$. Therefore, we always have  $(c_\alpha^{(\mathbf{p})}-c^{(\mathbf{p}_k)}_\alpha )\leq 0$ which implies that $$e^{(c_\alpha^{(\mathbf{p})}-c^{(\mathbf{p}_k)}_\alpha ) t}\leq 1$$ for all $t\geq 0$. As $x_k>0$ and $\sum_{k=1}^K x_k=1$, Eq. \eqref{eq:identity}   holds only if 
	$$e^{(c_\alpha^{(\mathbf{p})}-c^{(\mathbf{p}_k)}_\alpha ) t} =1$$ 
	for all $t\geq 0$. Therefore, this forces $c_\alpha^{(\mathbf{p})}=c_\alpha^{(\mathbf{p}_k)}$ for all  $\alpha\in [d+1]$ and $k\in [K]$	 from which  one deduces that 
	$\Lambda_t^{(\mathbf{p}_k)}=\Lambda_t^{(\mathbf{p})}$ for all $k\in [K]$. Therefore, we obtain a contradiction as we have assumed that  $\Lambda_t^{(\mathbf{p}_k)}$ are different. 
	
\end{proof}

\begin{proposition}\label{pro:genPauli2_3}
We always have the inclusion $\mathcal{S}_2\subseteq \mathcal{D}_2$ but	$\mathcal{S}_d\not\subseteq \mathcal{D}_d$ for all $d\geq 3$.
\end{proposition}
\begin{proof} 
 
 First, we consider the qubit case, that is, $d=2$. By the proof of Theorem \ref{thm:GPauli}, we known that the probability distribution functions of any Markovian semigroup  $\Lambda_t^{(\mathbf{p})}$ are of the form
 $$ p_\beta(t)=\frac{1}{4}(1+2e^{-c_\beta t}-\sum_{\alpha=1}^3 e^{-c_\alpha t}).$$
 Without loss of generality, we assume that $0\leq c_1\leq c_2\leq c_3.$ Therefore,
 $$ \begin{array}{c}
 	0\leq 1+e^{-c_1 t}-e^{-c_2 t}-e^{-c_3 t}< 1+e^{-c_1 t}\leq 2,\\ [3mm]
 	0\leq 1+e^{-c_2 t}-e^{-c_3 t}-e^{-c_1 t}\leq 1-e^{-c_3 t}<1,\\ [3mm]
 	0\leq 1+e^{-c_3 t}-e^{-c_1 t}-e^{-c_2 t}\leq 1-e^{-c_1 t}<1.  	
 \end{array}
 $$
 Therefore, 
 $ \sup_{t\geq 0}  p_1(t)+ \sup_{t\geq 0}  p_2(t) +\sup_{t\geq 0}  p_3(t) \leq 1.$
 Moreover, if we define
 $$ \begin{array}{l}
 	\pi_1 (t):=(1+e^{-c_1 t}-e^{-c_3 t}-e^{-c_2 t})/2, \\[3mm]
 	\pi_2 (t):=1+e^{-c_2 t}-e^{-c_3 t}-e^{-c_1 t}, \\[3mm]
 	\pi_3 (t):=1+e^{-c_3 t}-e^{-c_1 t}-e^{-c_2 t},       	   
 \end{array}
 $$
 then one can check that
 $$\Lambda_t^{(\mathbf{p})}=\frac{1}{2} \Lambda_t^{(\mathbf{p}_{1,\pi_1})}+\frac{1}{4} \Lambda_t^{(\mathbf{p}_{2,\pi_2})}+\frac{1}{4} \Lambda_t^{(\mathbf{p}_{3,\pi_3})}.$$

 Now we show the  
 surprise  relation:   $\mathcal{S}_d\not \subseteq \mathcal{D}_d$ for $d\geq 3.$ 
 In fact, by Theorem \ref{thm:GPauli}, we can define a series of Markovian semigroups  of generalized Pauli channels $\Lambda_t^{(\mathbf{p}_c)}$, which arising from the setting $c_1=c$ and $c_2=c_3=\cdots=c_{d+1}=1,$ here we assume that $0<c<1$ and 
 $\mathbf{p}_c:=(p_0^{(c)}(t),p_1^{(c)}(t),\cdots,p_{d+1}^{(c)}(t))$, where
 $$p^{(c)}_\beta(t)=\frac{d-1}{d^2}(1+d e^{-c_\beta t}-\sum_{\alpha=1}^{d+1} e^{-c_\alpha t}).$$
 
 By taking $t$ tend to infinite, we obtain that 
 $$\sup_{t\geq 0} p^{(c)}_\beta(t) \geq \frac{d-1}{d^2}$$
 for each $\beta\in [d+1].$ Moreover, note that 
 $$p^{(c)}_1(t)=\frac{d-1}{d^2}\left(1+(d-1) e^{-c  t}-d e^{- t}\right).$$
 Fix any $t>\log 2$,  equivalently,  $1-e^{-t} >1/2$ and let $c$ tend to zero, we have 
 $$\lim_{c\rightarrow 0^+} p_1^{(c)}(t)=\frac{d-1}{d}(1-e^{-t})>\frac{d-1}{2d}\geq \frac{1}{d},$$
 where the last inequality holds only if $d\geq 3.$
 Therefore, there exists a small enough  $c>0$  such that 
 $\sup_{t\geq 0} p^{(c)}_1(t) > \frac{1}{d}.$
 For this $c$, we have 
 $$ \sum_{\alpha=1}^{d+1}\left(\sup_{t\geq 0}  p^{(c)}_\alpha(t) \right) >\frac{1}{d}+d\times\frac{d-1}{d^2}=1.$$
 By Proposition \ref{pro:genPauli}, the Markovian semigroup $\Lambda_t^{(\mathbf{p}_c)}$ do not belong to $\mathcal{D}_d.$
 \end{proof}

The condition $\sum_{\alpha\neq \beta} \gamma_\alpha(t)\geq 0$ for all $\alpha\in [d+1]$ is a sufficient and necessary condition for a dynamical map $\Lambda_t^{(\mathbf{p})}$ to be $P$-divisible when $d=2$. However, this does not hold for $d\geq 3$ \cite{Wud2015}  which provides a fundamental difference between the qubit and general qudit cases.
In the following, we will show that $\mathcal{S}_2\subseteq \mathcal{D}_2$ but $\mathcal{S}_d\not\subseteq \mathcal{D}_d$ for all $d\geq 3$, which provides another difference between the qubit and general qudit cases.
  Therefore, the relations among $\mathcal{P}_2, \mathcal{D}_2$ and $\mathcal{S}_2$ and relations among $\mathcal{P}_d, \mathcal{D}_d$ and $\mathcal{S}_d$ for  $d\geq 3$   can be viewed in Fig. \ref{fig:pauli}, respectively.
 
 \begin{figure}[ht]
 	\centering
 	\includegraphics[scale=0.9]{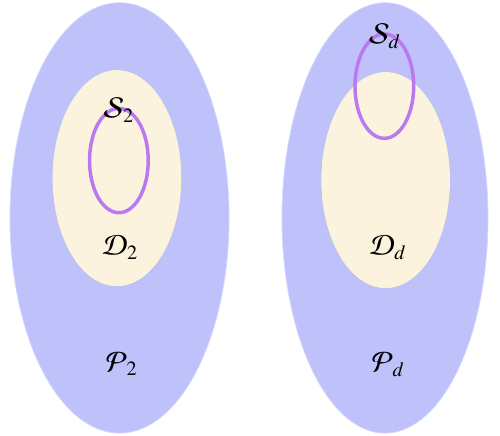}
 	\caption{The relations between the three sets $\mathcal{P}_2,\mathcal{D}_2$ and $\mathcal{S}_2$ and relations between the three sets $\mathcal{P}_d,\mathcal{D}_d$ and $\mathcal{S}_d$  for $d\geq 3$. }\label{fig:pauli}
 \end{figure}

It is interesting to note that by convex combination of Markovian semigroups in $\mathcal{S}_d$, the resultant map could has $(d-1)^2$ local decoherence rates that are permanently negative \cite{Siudzinska_JPA20}. However, we do not know whether $(d-1)^2$ is the largest number of permanently negative decoherence rates.
In the following, we prove that this is true for the setting of generalized Pauli channels. First, let $\{\gamma_\alpha(t)\}_{\alpha\in [d+1]}$ be the local decoherence rates of the resultant map of a convex combination of Markovian semigroups in $\mathcal{S}_d$.

\begin{proposition}\label{prop:mixSemiPauli}
	Given $d\geq2$ is an integer	 and suppose that  $\Lambda_t^{(\mathbf{p})}$ is a convex combination of  $K$ Markovian  semigroups $\{\Lambda_t^{(\mathbf{p}_k)}\}_{k=1}^K \subseteq \mathcal{S}_d,$ that is, 
	$$ \Lambda_t^{(\mathbf{p})}=\sum_{k=1}^K x_k \Lambda_t^{(\mathbf{p}_k)},  \text{ where }  x_k\geq 0 \text{ and }  \sum_{k=1}^K x_k =1.$$ 
	If the  time local generator of $\Lambda_t^{(\mathbf{p})}$ is 
	$$\mathcal{L}_t^{(\mathbf{p})}[\rho]=\frac{1}{d}\sum_{\alpha=1}^{d+1} \gamma_\alpha(t)\left(\mathbb{U}_\alpha [\rho]-(d-1) \rho\right),$$ 
	then  for all $\beta\in [d+1],$ we have
	$$
	\sum_{\alpha\neq \beta} \gamma_\alpha(t)\geq 0.$$
	In particular, for $d=2$, the mixture of Pauli channels must be $P$-divisible.
\end{proposition} 
\begin{proof}
	The eigenvalues of $\Lambda_t^{(\mathbf{p}_k)}$ can be assumed to be $e^{-c_j^{(k)}t}$, where $(j,k)\in [d+1] \times [K]$ and $c_j^{(k)}\geq 0$.  So the eigenvalues of $\Lambda_t^{(\mathbf{p})}$ are 
	$$\lambda_j(t)=\sum_{k=1}^K x_ke^{-c_j^{(k)}t}>0.$$
	Therefore,   $\Lambda_t^{(\mathbf{p})}$ is invertible and 
	$$
	\dot{\Lambda}_t^{(\mathbf{p})} (\Lambda_t^{(\mathbf{p})})^{-1} [U_{\beta}^{\ell}]=\mathcal{L}_t^{(\mathbf{p})} [U_{\beta}^{\ell}].$$
	The left hand side is $\dot{\lambda}_\beta(t)/\lambda_\beta(t) U_{\beta}^{\ell}$, while the right hand side is $-\sum_{\alpha \neq \beta } \gamma_\alpha(t) U_{\beta}^{\ell}.$ Therefore,
	$$ \sum_{\alpha\neq \beta} \gamma_\alpha(t)= -\dot{\lambda}_\beta(t)/\lambda_\beta(t).$$
	By the definition of $\lambda_\beta(t)$, we have 
	$$\dot{\lambda}_\beta(t) =-\sum_{k=1}^K c_\beta^{(k)} x_k e^{-c_\beta^{(k)}t}\leq 0.$$
	Therefore, $\sum_{\alpha\neq \beta} \gamma_\alpha(t)\geq 0.$
\end{proof}

 From Proposition \ref{prop:mixSemiPauli}, one finds that there are no $d$ terms among $\{\gamma_\alpha(t)\}_{\alpha\in [d+1]}$ that are simultaneously negative at some times $t\geq  0$.   As a consequence, there are at most  $(d-1)^2$ local decoherence rates that are permanently negative. Moreover, we have the following statement.

\begin{proposition}\label{pro:genPauli_extenal}
	Each generalized Pauli channel in $\mathcal{P}_d$ could have  at most $(d-1)^2$ local decoherence rates that are permanently negative for all $t\geq0$.
\end{proposition}
\begin{proof}
	Let $\Lambda_t^{(\mathbf{p})} \in \mathcal{P}_d$ and suppose that the local decoherence rates are $\{\gamma_\alpha(t)\}_{\alpha=1}^{d+1}$ (each is $(d-1)$ folds). It is sufficient to prove that among   $\{\gamma_\alpha(t)\}_{\alpha=1}^{d+1}$, there are at most $(d-1)$ terms could be permanently negative for all $t\geq0$. Suppose not, without loss of generality, we could assume that $\{\gamma_\alpha(t)\}_{\alpha=2}^{d+1}$ are all permanently negative for all $t\geq0$.	  From the  equation
	$
	\dot{\Lambda}_t^{(\mathbf{p})}  = {\Lambda}_t^{(\mathbf{p})} \circ {\mathcal{L}}_t^{(\mathbf{p})},$
	one finds that 
	$$\dot{\lambda}_\alpha(t)=-\left(\sum_{\beta\neq \alpha } \gamma_\beta(t)\right ) \lambda_\alpha(t), \ \ \forall \alpha\in [d+1].$$
	Define $\gamma(t)=\sum_{\beta=2}^{d+1} \gamma_\beta(t).$ By the assumption,   we have $\gamma(t)<0$ for all $t\geq0$. Therefore, solving the equation $\dot{\lambda}_1(t)=-\gamma(t) \lambda_1(t),  \lambda_1(0)=1,$ one has 
	$$\lambda_1(t)=\mathrm{exp}\left[-\int_0^t \gamma(\tau)\mathrm{d} \tau \right]>1, \forall t\geq0$$
	which is contradicted with the well-known constraints $|\lambda_\alpha(t)|\leq 1.$
\end{proof}

 In Ref. \cite{Jagadish2022}, the authors demonstrated that non-invertibility is a prerequisite for generating a Markovian semigroup when mixing $(d+1)$ generalized dephasing channels. However, in the subsequent discussion, we reveal that this non-invertibility requirement can be eliminated by substituting the previous mixing channels with more general generalized Pauli channels from the set $\mathcal{P}_d$. In fact, for any integer $n\geq 2$, we can establish the existence of a Markovian semigroup that can be represented as a convex combination of $n$ distinct invertible generalized Pauli channels within $\mathcal{P}_d$.

We consider the simplest Markovian semigroup $\Lambda_t^{(\mathbf{p}) }$: $c_1=c_2=\cdots=c_{d+1}=1$ and $p_\alpha(t)=\frac{d-1}{d^2}(1-e^{-t})$ for $\alpha\in [d+1].$ First, we consider the case $n=2$.
For each $\alpha\in [d+1]$, we define 
$$
\begin{array}{l}
\displaystyle p_\alpha^{(1)}(t)=p_\alpha(t)(1-e^{-t})=\frac{d-1}{d^2}(1-e^{-t})^2, \\ [2mm]
\displaystyle p_\alpha^{(2)}(t)=p_\alpha(t)(1+e^{-t})=\frac{d-1}{d^2}(1-e^{-2t}). \\ [2mm]
\end{array}
$$ Therefore, we obtain two generalized Pauli channels $\Lambda_t^{(\mathbf{p}_1) }$ and $\Lambda_t^{(\mathbf{p}_2) }$, which satisfy $\Lambda_t^{(\mathbf{p}) }=\frac{1}{2} \Lambda_t^{(\mathbf{p}_1) }+\frac{1}{2} \Lambda_t^{(\mathbf{p}_2) }.$ 
Moreover, the eigenvalues of  $\Lambda_t^{(\mathbf{p}_1) }$ and $\Lambda_t^{(\mathbf{p}_2) }$ are 
$$
\begin{array}{l}
	\displaystyle	\lambda_\alpha^{(1)}(t)= 1-\frac{d^2}{d-1} p_\alpha^{(1)}(t)=1-(1-e^{-t})^2, \\ [2mm]
	\displaystyle	\lambda_\alpha^{(2)}(t)= 1-\frac{d^2}{d-1} p_\alpha^{(2)}(t)=1-(1-e^{-2t})=e^{-2t},  \\ [2mm]
\end{array}
$$ 
respectively,  which are both greater than 0 for all $t\geq 0$. Therefore, the Markovian semigroup $\Lambda_t^{(\mathbf{p})}$ is a convex combinations of two invertible channels. Note that $\Lambda_t^{(\mathbf{p}_2)}$ is again a Markovian semigroup. 

For the case of $n=3$,  we  first show that $\Lambda_t^{(\mathbf{p}_2) }$ can be decomposed into a convex combination of two invertible generalized Pauli channels.  For each $\alpha\in [d+1]$, we define 
$$
\begin{array}{l}
	\displaystyle q_\alpha^{(1)}(t)=p^{(2)}_\alpha(t)(1-e^{-2t})=\frac{d-1}{d^2}(1-e^{-2t})^2, \\ [2mm]
	\displaystyle q_\alpha^{(2)}(t)=p^{(2)}_\alpha(t)(1+e^{-2t})=\frac{d-1}{d^2}(1-e^{-4t}). \\ [2mm]
\end{array}
$$  
 Therefore, we obtain two generalized Pauli channels $\Lambda_t^{(\mathbf{q}_1) }$ and $\Lambda_t^{(\mathbf{q}_2) }$, which satisfy $\Lambda_t^{(\mathbf{p}_2) }=\frac{1}{2} \Lambda_t^{(\mathbf{q}_1) }+\frac{1}{2} \Lambda_t^{(\mathbf{q}_2) }.$ 
Moreover, the eigenvalues of  $\Lambda_t^{(\mathbf{q}_1 )}$ and $\Lambda_t^{(\mathbf{q}_2) }$ are 
$$
\begin{array}{l}
	\displaystyle	\mu_\alpha^{(1)}(t)= 1-\frac{d^2}{d-1} q_\alpha^{(1)}(t)=1-(1-e^{-2t})^2, \\ [2mm]
	\displaystyle	\mu_\alpha^{(2)}(t)= 1-\frac{d^2}{d-1} q_\alpha^{(2)}(t)=1-(1-e^{-4t})=e^{-4t},  \\ [2mm]
\end{array}
$$ 
respectively, which are both greater than 0 for all $t\geq 0$.  Therefore, the Markovian semigroup $\Lambda_t^{(\mathbf{p}_2)}$ is a convex combinations of two invertible channels $\Lambda_t^{(\mathbf{q}_1)}$ and $\Lambda_t^{(\mathbf{q}_2)}$. Note that $\Lambda_t^{(\mathbf{q}_2)}$ is again a Markovian semigroup. 
Therefore,  $\Lambda_t^{(\mathbf{p}) }$ can be decomposed as a convex combination of three diferent invertible generalized Pauli channels
$$\Lambda_t^{(\mathbf{p}) }=\frac{1}{2} \Lambda_t^{(\mathbf{p}_1) }+\frac{1}{2} \Lambda_t^{(\mathbf{p}_2) }=\frac{1}{2} \Lambda_t^{(\mathbf{p}_1) }+\frac{1}{4} \Lambda_t^{(\mathbf{q}_1) }+\frac{1}{4} \Lambda_t^{(\mathbf{q}_2) }.$$ 
We can follow the same procedure as described above and discover that $\Lambda_t^{(\mathbf{p})}$ can be expressed as a mixture of $n$ distinct invertible generalized Pauli channels. This finding diverges from the conventional understanding that non-invertibility is a prerequisite for generating a Markovian semigroup through the convex combination of generalized Pauli dephasing channels.

 \section{Conclusions}\label{four}



  In our investigation, we delved into the intricate properties arising from the convex combination of generalized Pauli channels. Our exploration led to several noteworthy findings.
First, we made a surprising discovery: the existence of certain generalized Pauli channels that could not be represented as a mere convex combination of $(d+1)$ generalized Pauli dephasing channels. This revelation challenged our intuition, which had been shaped by historical literature.   Subsequently, we endeavored to broaden our understanding beyond the confines of mixing just $(d+1)$ generalized Pauli dephasing channels. Instead, we examined the more general scenario of mixing any generalized Pauli channels.  Remarkably, many fundamental properties remained consistent in this expanded framework. For instances, any mixture of invertible generalized Pauli  channels is still invertible; any nontrivial convex combination of generalized Pauli  channels which are Markovian semigroups  could not lead to a  Markovian semigroup again. 

Moreover, we present a sufficient and necessary  condition for a generalized Pauli  channel to be a Markovian semigroup   via the spectra of the dynamical map.  This criterion allowed us to demonstrate that every Pauli channel (for $d=2$) could be expressed as a mixture of $(d+1)$ Pauli dephasing channels. However, this generality did not extend to higher dimensions, highlighting a crucial distinction between qubits and general qudits.

Additionally, we unveiled that each generalized Pauli channel of dimensionality $d$ could exhibit, at most, $(d-1)^2$ number of local decoherence rates that remained persistently negative for all time intervals $t\geq0$. This finding shed light on the constraints governing these channels when subjected to decoherence effects. We discovered that by combining certain invertible generalized Pauli channels, we can create  a Markovian semigroup. This outcome is different from what was previously reported in the research by Jagadish et al.  \cite{Jagadish2022}.

It is interesting to study the Markovian and non-Markovian properties of the resultant map under  mixing the generalized  Pauli  channels. These intriguing questions prompt further exploration in the field. Additionally, we wonder whether it is possible to construct a Weyl channel of dimensionality $d$ that exhibits more than $(d-1)^2$ local decoherence rates that remain permanently negative for all time intervals $t\geq0$.

\section*{Acknowledgements}

This work is supported by
National Natural Science Foundation of China
(12371458, 11901084),   the Key Research and Development Project of Guangdong province under Grant No. 2020B0303300001, the Guangdong Basic and Applied Research Foundation under Grant No.  2023A1515012074 and 2020B1515310016, Key Lab of Guangzhou for Quantum Precision Measurement under Grant No. 202201000010, the Science and
Technology Planning Project of Guangzhou under
Grants No. 2023A04J1296.

\appendix

\section{A fundamental lemma}\label{sec:AppendixA}

\begin{lemma}\label{lemm：zero}
	Given an integer $d\geq 2$ and $d^2$ linearly independent operators $V_{ij}\in  \mathbb{L}_d$ where $(i,j)\in [d] \times [d]$. For any $\mathbf{x}:=(x_{ij})_{(i,j)\in [d]\times [d]}$ where $x_{ij}\in \mathbb{C}$, we can define a map from $\mathbb{L}_d$ to $\mathbb{L}_d$ given by  $$\Lambda^{(\mathbf{x})}[\sigma]:= \sum_{i=1}^d\sum_{j=1}^d x_{ij} V_{ij} \sigma V_{ij}^\dagger,\ \  \forall \sigma \in \mathbb{L}_d. $$
	If $\Lambda^{(\mathbf{x})}$ 
	 acts trivially on all density matrices $\mathbb{D}_d$,  that is, $\Lambda^{(\mathbf{x})}[\rho]=\mathbf{0}$ for all $\rho\in \mathbb{D}_d$, then   $\mathbf{x}=\mathbf{0}.$
	\end{lemma}
\begin{proof}
By the linearity of $\Lambda^{(\mathbf{x})}$ and well-known result 
	$ \mathrm{span}_\mathbb{C}(\mathbb{D}_d)=\mathbb{L}_d,$  the operation $\Lambda^{(\mathbf{x})}$ also acts trivially on all matrices in  $\mathbb{L}_d,$ that is,  $\Lambda^{(\mathbf{x})}$ is the zeros operation from  $\mathbb{L}_d$ to $\mathbb{L}_d.$
	Denote $\{|e_i\rangle\}_{i\in [d]}$ the computational basis of $\mathcal{H}_d$. Then we have $$\mathbf{0}=\Lambda^{(\mathbf{x})}(|e_k\rangle \langle e_l|)=\sum_{(i,j)\in [d] \times [d]} x_{ij} V_{ij}  |e_k\rangle \langle e_l| V_{ij}^\dagger,$$
	which is equivalent to 
	$$\sum_{(i,j)\in [d] \times [d]} x_{ij} V_{ij}\otimes  V_{ij}^*  |e_k\rangle\otimes |e_l\rangle=\mathbf{0}$$
	for all $(k,l)\in [d]\times [d]$. Hence $\sum_{(i,j)\in [d] \times [d]} x_{ij} V_{ij}\otimes  V_{ij}^*$ is a zero map from $\mathcal{H}_d\otimes \mathcal{H}_d$ to itself. By the tensor theory and the linear independence of $\{V_{ij}\}$, we know that 
	the set $\{V_{ij}\otimes V^*_{kl}\}_{i,j,k,l=1}^d$ form a basis of $\mathbb{L}_d\otimes \mathbb{L}_d.$ Therefore, the set $\{V_{ij}\otimes V^*_{kl}\}_{i,j,k,l=1}^d $ is linearly independent which implies that $x_{ij}=0$ for all $(i,j)\in [d] \times [d].$
	\end{proof}

\begin{corollary}\label{cor:channelequal}
	Given an integer $d\geq 2$ and two probability distribution functions $\mathbf{p}=(p_0(t),p_1(t),\cdots,p_{d+1}(t))$ and  $\mathbf{q}=(q_0(t),q_1(t),\cdots,q_{d+1}(t))$. The two generalized Pauli channels $\Lambda_t^{(\mathbf{p}) }$ and  $\Lambda_t^{(\mathbf{q}) }$ are equal if and only if $\mathbf{p}=\mathbf{q}.$
	
	\end{corollary}

This could be deduced easily from  Lemma \ref{lemm：zero} once  one notes that the $d^2$ unitary matrices $\{\mathbb{I}_d\}\cup\{U_{\alpha}^k\}_{(\alpha,k)\in [d+1]\times [d-1]}$  (where $U_{\alpha}$ is defined by Eq. \eqref{eq:UU}) arising from the ($d+1$) MUBs are linearly independent.

 \section{Proof of Example \ref{exam:non=invertible}}\label{sec:appB}
 Note that  for any general probabilities distribution $\mathbf{p},$ the eigenvalue equations of $\mathcal{E}_t^{(\mathbf{p})} $ defined by Eq. \eqref{eq:Weyl}  always  satisfy  $  	\mathcal{E}_t^{(\mathbf{p})}(U_{kl})=\lambda_{kl}(t)U_{kl}, \forall (k,l)\in\mathbb{Z}_{d}\times\mathbb{Z}_{d}, $
where the time-dependent eigenvalues $\lambda_{kl}(t)$ could be expressed as 
\begin{equation}\label{eq:eigexpression}
	\lambda_{kl}(t)=\sum\limits_{(i,j)\in\mathbb{Z}_{d}\times\mathbb{Z}_{d}}H_{ij,kl}\; p_{ij}(t),
\end{equation}
with $H$ being the $d^{2}\times d^{2}$ Hermitian complex matrix defined by $H_{ij,kl}=\omega_d^{jk-il}$.

\vskip 5pt

\noindent {\bf Proof of Example \ref{exam:non=invertible}.} 
 Set $\omega_3=e^{\frac{2\pi \mathrm{i}}{3}}$. Denote $\mathbf{r}=\frac{1}{2} \mathbf{p}+\frac{1}{2} \mathbf{q}$, that is, the probabilities distribution of the channel $\frac{1}{2} \mathcal{E}_t^{(\mathbf{p})}+\frac{1}{2} \mathcal{E}_t^{(\mathbf{q})}.$ Suppose their eigenvalue equations are 
 $
	\mathcal{E}_t^{(\mathbf{x})}(U_{kl})=\lambda^{(\mathbf{x})}_{kl}(t)U_{kl}, \forall (k,l)\in\mathbb{Z}_{3}\times\mathbb{Z}_{3}, $
for $\mathbf{x}\in \{\mathbf{p},\mathbf{q},\mathbf{r}\}.$  By Eq. \eqref{eq:eigexpression}, we have 
$$
\begin{array}{rl}
   \lambda^{(\mathbf{p})}_{kl}  &=\left( \displaystyle\sum_{i=0}^2  \omega^{-il} p_{i}(t)\right) \left( \displaystyle\sum_{j=0}^2  \omega^{jk} p_{j}(t)\right),   \\[5mm]
    \lambda^{(\mathbf{q})}_{kl}  &=\left( \displaystyle\sum_{i=0}^2  \omega^{-il} q_{i}(t)\right) \left( \displaystyle\sum_{j=0}^2  \omega^{jk} p_{j}(t)\right),\\[5mm]
    \lambda^{(\mathbf{r})}_{kl}  &=\left( \displaystyle\sum_{i=0}^2  \omega^{-il} \frac{p_i(t)+q_{i}(t)}{2}\right) \left( \displaystyle\sum_{j=0}^2  \omega^{jk} p_{j}(t)\right).
\end{array}
$$
Note that  $\displaystyle\sum_{i=0}^2  \omega^{-il} p_{i}(t)=1$ for $l=0$ and 
$$
\begin{array}{c}\displaystyle\sum_{i=0}^2  \omega^{-i} p_{i}(t) =p_0(t)-\frac{p_1(t)+p_2(t)}{2} +\mathrm{i} \frac{\sqrt{3}}{2}(p_2(t)-p_1(t)), \\ 
\displaystyle\sum_{i=0}^2  \omega^{-2i} p_{i}(t) =p_0(t)-\frac{p_1(t)+p_2(t)}{2} +\mathrm{i} \frac{\sqrt{3}}{2}(p_1(t)-p_2(t)), 
\end{array}$$
which are all nonzero  as $p_1(t)\neq p_2(t)$ for all $t> 0$. Similarily,
$\displaystyle\sum_{i=0}^2  \omega^{-il} q_{i}(t)\neq 0$ for all $t>0$.  Note that $\sum_{j=0}^2  \omega^{jk} p_{j}(t)$ is just the complex conjugate of $\sum_{j=0}^2  \omega^{-jk} p_{j}(t).$ Therefore, both $\lambda^{(\mathbf{p})}_{kl} $ and  $\lambda^{(\mathbf{q})}_{kl} $  are nonzero for all $(k,l) \in \mathbb{Z}_3\times \mathbb{Z}_3.$ That is, both $\mathcal{E}_t^{(\mathbf{p})}$ and $\mathcal{E}_t^{(\mathbf{q})}$  are invertible. Moreover, one finds that 
$$ \displaystyle\sum_{i=0}^2  \omega^{-i} \;\frac{p_i(t)+q_{i}(t)}{2}=p_0(t)-\frac{p_1(t)+p_2(t)}{2}=\frac{5 e^{-t}-1}{4},$$ 
which is equal to zero when $t= \log 5.$ Therefore,  the channel $\frac{1}{2} \mathcal{E}_t^{(\mathbf{p})}+\frac{1}{2} \mathcal{E}_t^{(\mathbf{q})}$ is non-invertible. \qed

\end{document}